\documentclass[12pt,epsfig]{article}

\usepackage[margin=1in]{geometry}
\usepackage{amsfonts}
\usepackage{amsmath}
\usepackage{amssymb}
\usepackage{amsthm}
\usepackage{graphicx}
\usepackage{epsfig}
\usepackage{color}

\numberwithin{equation}{section}

\newtheorem{theorem}{Theorem}
\newtheorem{lemma}[theorem]{Lemma}

\newtheorem{corollary}[theorem]{Corollary}
\newtheorem{proposition}[theorem]{Proposition}
\newtheorem{example}{Example}

\newcommand{\thmref}[1]{Theorem~{\rm \ref{#1}}}
\newcommand{\lemref}[1]{Lemma~{\rm \ref{#1}}}
\newcommand{\corref}[1]{Corollary~{\rm \ref{#1}}}
\newcommand{\propref}[1]{Proposition~{\rm \ref{#1}}}

\newcommand{\exmref}[1]{Example~{\rm \ref{#1}}}

\title{Approximating Functional of Local Martingale Under the Lack of
  Uniqueness of Black-Scholes PDE 
}
\begin{document}

\author{
  Qingshuo Song \thanks{Department of Mathematics, City University of
    Hong Kong, {\tt song.qingshuo@cityu.edu.hk}, The research of
    Q.S. is 
    supported in part by the Research Grants Council of Hong Kong
    No. CityU 103310. }  }

\maketitle

\begin{abstract}
 When the underlying stock
  price is a strict local
  martingale process  under an equivalent local martingale measure,   
  Black-Scholes PDE associated with an European option may have
  multiple   solutions. In this paper, we   study an approximation for 
  the  smallest hedging price of such an European option. Our
  results show that a class of rebate barrier
  options can be used for this approximation. Among of them, a
  specific rebate option is also provided 
  with a continuous rebate function, which corresponds to the unique 
  classical solution of the associated parabolic PDE. Such a
  construction makes existing  numerical PDE techniques applicable
  for its computation.     An asymptotic convergence   rate is   also 
  studied when the knocked-out barrier moves to  infinity under
  suitable   conditions.      
\end{abstract}

%\noindent {\bf AMS subject classifications.} 78M20, 65N22, 65N06


\section{Introduction}\label{intro}

In a financial market equipped with the unique equivalent local martingale
measure (ELMM) $\mathbb{P}$, the smallest hedging price of
an European option is the conditional
expectation  of the payoff with respect to the probability
$\mathbb{P}$, see \cite{FK10, JYC09}. 
In contrast to the probabilistic representation, option price can be also characterized as the
unique solution of its associated Black-Scholes PDE,
provided that PDE has a unique classical solution.

The necessary and sufficient condition for the unique solvability of
the parabolic PDE is that, the underlying stock price
 is Martingale process, see \cite{BX10}. In
other words, if the stock price is a strict local Martingale, then there exists
multiple solutions for Black-Scholes PDE. Moreover, the option price may 
be  one of the many solutions, see \cite{ET09}. The
difference of the multiple solutions of PDE is termed as financial
bubbles, see \cite{CH05, ET09, JPS07} and the references therein.
In this work, we will consider the following problem proposed by
Fernholz and Karatzas \cite{FK10}:
\begin{itemize}
\item [(Q)] {\it How can one find a feasible
  numerical solution convergent  to the 
  option price under the lack of uniqueness of
  Black-Scholes PDE?} 
\end{itemize}

We first examine the existing numerical schemes on CEV model of
Example~\ref{exm:cev}, where the option price can be explicitly
identified. There are typically two kinds of numerical schemes in this
vein \cite{Wil07}. One is Monte Carlo method by discretizing the
probability representation, the other is PDE numerical
method by discretizing the truncated version of PDE. 

Unfortunately, Example~\ref{exm:EM} and Example~\ref{exm:FDM} shows
that classical Euler-Maruyama approximation (for Monte Carlo method)
and finite  difference method (for PDE numerical method) leads to a
strictly larger value than the desired option price. 
Motivated from these two examples, question (Q) boils down
into the following two problems: 
\begin{itemize}
\item [(Q1)] Find a feasible approximation for a Monte Carlo method, and
  its convergence rate;
\item [(Q2)] Find a feasible approximation for PDE numerical method,
  and its convergence rate.
\end{itemize}

In short, this work intends to find a feasible approximation to  the
smallest superhedging price $V(x,t)$. It turns out
that the value function 
can be obtained by a limit of a series of appropriate rebate option
prices, which can be estimated by usual Monte Carlo method, 
see the details in \corref{cor:mc}. 
However, Corollary~\ref{cor:mc} may not be utilized for 
the approximation by PDE numerical method, since it may cause a discontinuity at the corner of the terminial-boundary datum.
Therfore, a specific  rebate option is proposed with its price
continuous up to the boundary, so that its
price corresponds to the unique classical solution of its associated
parabolic PDE. Such a construction makes existing  numerical PDE
techniques applicable  for computations, see Theorem~\ref{thm:contbd}.

The rest of the paper is outlined as follows. In the next section, we
give precise formulation of the problem. Section 3 presents main
results, and related proofs is relegated to Section 4 for the reader's convenience. The last section summarizes the work.

\section{Problem formulation}

Throughout this paper, we use $K$ as a generic constant, and
$\mathbb{R}^+ = (0,\infty)$, $\mathbb{\bar R}^+ = 
\mathbb{R}^+ \cup \{0\}$.
If $A$ is a subset of $\mathbb{R}\times [0,T]$, then  $C(A)$ denotes
the set of all continuous real functions on $A$,  $C^{2,1}(A)$
denotes a collection of all functions $\varphi: A \mapsto \mathbb{R}$
such that $\varphi_{xx}$ and $\varphi_t$  belong to $C(A)$.
$D_\gamma(A)$ denotes the set of all measurable  functions $\varphi: A
\to \mathbb{\bar R^+}$ satisfying growth condition
\begin{equation}
  \label{eq:grow}
  \varphi(x,t) \le K(1+|x|^\gamma), \ \forall (x,t) \in A.
\end{equation}
$C_\gamma(A) = C(A) \cap D_\gamma(A)$ denotes the set of all
continuous functions satisfying $\gamma$-growth.
We also denote the parabolic domain $Q:= \mathbb{R}^+ \times
(0,T)$,  truncated domain $Q_\beta:= (0,\beta) \times (0,T)$, and
$Q_\beta^\alpha:= (\alpha,\beta) \times (0,T)$ for $0<\alpha<\beta$.

We consider a single stock in the presence of the unique equivalent
local 
martingale measure (ELMM) $\mathbb{P}$, under which the deflated price
process follows  
\begin{equation}
  \label{eq:sde}
  d X(s) = \sigma (X(s)) dW(s), \ X(t) = x \ge 0,
\end{equation}
where  $W$ is a standard
Brownian motion with respect to a given probability space $(\Omega,
\mathcal{F}, \mathbb{P}, \mathbb{F} = \{\mathcal{F}_s: s\ge t\})$
satisfying usual conditions.  
We impose the following two conditions on $f$ and $\sigma$: 
  \begin{enumerate}
  \item [(A1)]
    $\sigma$ is locally H\"older continuous with exponent $\frac 1 2$
    satisfying 
    $\sigma(x)>0$ for all $x\in \mathbb{R}^+$,  $\sigma(0) = 0$. 
  \item [(A2)] $f: \mathbb{\bar R}^+  \to \mathbb{\bar R}^+$ is
    a $C_\gamma(\mathbb{\bar R}^+)$  
    payoff function for some $\gamma \in [0,1]$.  
  \end{enumerate}
By \cite[5.5.11]{KS98},  the assumption (A1) on $\sigma$ ensures there
exists a unique strong solution of \eqref{eq:sde} with absorbing state
at zero.

For a contingent claim $f(X(T))$ with a fixed maturity $T>0$, the
smallest hedging price has the form of
\begin{equation}
  \label{eq:V}
  V(x,t) = \mathbb{E}_{x,t}[f(X(T))] := \mathbb{E}[f(X^{x,t}(T)) |
  \mathcal{F}_{t}]. 
\end{equation}
In the above, we suppress the superscripts $(x,t)$ in
$X^{x,t}$, and write $\mathbb{E}_{x,t}[\ \cdot\ ]$ to indicate the
expectation with respect to $\mathbb{P}$ computed under these initial
conditions.  

Recently, \cite{ET09}
shows that the value function $V$ of \eqref{eq:V} is  the
$C^{2,1}(Q)\cap C(\overline Q)$ solution of    $BS(Q, f)$, where
$BS(Q, f)$ refers to Black-Scholes equation
\begin{equation}
  \label{eq:bsq}
  BS(Q, f) \left\{
    \begin{array}{ll}
      (E) & \displaystyle u_t + \frac 1 2 \sigma^2(x) u_{xx} = 0
      \hbox{ on } Q =   \mathbb{R}^+ \times (0,T) \\
      (TD) &   \displaystyle u(x,T) = f(x)  \hbox{ on }  \forall x\in
      (0,\infty) \\ 
      (BD) & \displaystyle u(0,t) = f(0) \hbox{ on } \forall t\in
      (0,T]. 
    \end{array} \right. 
\end{equation}

However, the next example taken from \cite{CH05} shows that the value
function $V$ may not be the unique solution of $BS(Q,f)$ when the
deflated price process $X$ is a strict local martingale. 
\begin{example}[CEV model]\label{exm:cev}
  Suppose the stock price follows a strict local martingale
  process $dX(s) = X^2 dW(s)$, with the initial $X(t) =
  x>0$. Consider $V(x,t) = \mathbb{E}_{x,t}[X(T)]$. Then, 
  $V$ can be computed explicitly as
  \begin{equation}
    \label{eq:cev1}
    V(x,t) = x \Big(1 - 2 \Phi\Big(- \frac{1}{x\sqrt{T-t}}\Big) \Big).
  \end{equation}
  One can verify $V$ satisfies  $BS(Q,f)$. 
  Another trivial solution is $u(x,t) = x$. 
\end{example}

Now, $V$ is one of the possibly multiple solutions of $BS(Q,f)$.  With
the existence of multiple solutions to PDE $BS(Q,f)$, our question is
\begin{itemize}
\item If we use any of existing PDE numerical methods on
  $BS(Q,f)$, or any of existing Monte-carlo methods on \eqref{eq:V},
  does it converge to the desired value function among multiple
  solutions of PDE?
\end{itemize}

Unfortunately,  the answer is NO in general. In fact, the next
trivial example shows 
that the classical Monte Carlo method by Euler-Maruyama approximation
does {\it not} lead to the desired value $V(x,t)$ of \eqref{eq:cev1}
of \exmref{exm:cev}.  
\begin{example} \label{exm:EM}
  Consider the strong Euler-Maruyama (EM) approximation  to
  \exmref{exm:cev} with step size $\Delta$,
  $$X^\Delta_{n+1} = X_n^\Delta + \sigma(X_n^\Delta) (W(n\Delta +
  \Delta) - W(n\Delta)), \quad X^\Delta_t = x.$$
  Let $X^\Delta(\cdot)$ be the piecewise constant interpolation of
  $\{X^\Delta_n: n\ge 0\}$, i.e.
  \begin{equation}
    \label{eq:Xdelta}
    X^\Delta(s) = X^\Delta_{[s/\Delta]}, \quad \forall s>0. 
  \end{equation}
  Since $\{X^\Delta_n: n\ge 0\}$ is a martingale, the approximated value
  function simply leads to a wrong value
  $$V_\Delta(x,t) := \mathbb{E}_{x,t} [X^\Delta (T)] =
  \mathbb{E}_{x,t}[X^\Delta_{(T-t)/\Delta}] = x > V(x,t). \quad \Box$$ 
\end{example}

Similar to Monte-carlo method, One can also prove that the finite
difference method on PDE $BS(Q,f)$ also leads to a wrong value.
\begin{example} \label{exm:FDM}
  Black-Scholes PDE associated to the CEV model \exmref{exm:cev} is
  $BS(Q,f)$ of  \eqref{eq:bsq} with $\sigma(x) = x^2$ and $f(x) = x$.
  To use the finite difference method (FDM)  in the above
  PDE, we shall truncate the domain and put artificial boundary
  conditions on the upper barrier $\{(x,t): 0\le t \le T\}$ for large
  enough $\beta>0$. 
  As suggested by \cite{Wil07}, we impose boundary
  conditions, which {\it asymptotes the option price}, i.e.
  $$u(\beta,t) = \beta, \ \forall 0\le t\le T.$$
  With step size $\Delta^2$ in space variable $x$ and $\Delta$ in time 
  variable $t$, one can easily check upward finite difference scheme
  backward in time yields trivial numerical solution $u^\Delta(x,t) =
  x$ for any small $\Delta>0$. $\Box$
\end{example}

To the end, our work is to resolve the following question: 
How can one find a feasible approximation of this value function $V$
of \eqref{eq:V} in both Monte Carlo method and FDM? What is the
convergence rate?

\section{Main result}

In this subsection, we present the main results, and the proofs will
be relegated to the next section. 

\subsection{Approximation by Monte Carlo method}
We consider the following up-rebate option prices:
Suppose the up barrier is given by a positive constant
$\beta>x>0$ and  stopping time $\tau^\beta$ (suppressing the initial condition
$(x,t)$) is the
first hitting time of the stock price $X(s)$ to the barrier
$\beta$, i.e.
\begin{equation}
  \label{eq:taubeta}
  \tau^{x,t,\beta} = \inf\{s>t: X^{x,t}(s) \ge \beta \} \wedge T.
\end{equation}
For some function $g$, let its payoff at
$\tau^{x,t,\beta}$ consist of  
\begin{enumerate}
\item rebate payoff $g(\beta)$, if $\tau^\beta<T$;
\item otherwise,  terminal payoff $f(X(T))$.
\end{enumerate}
Then, the rebate option price $V^\beta$ is of the form 
\begin{equation}
  \label{eq:appr1}
  V^{\beta}(x,t) = \mathbb{E}_{x,t} [
  g(\beta) {\bf 1}_{\{ \tau^\beta <  T\}} +
  f(X(T)) {\bf  1}_{\{\tau^\beta = T\}}],
\end{equation}
$V^\beta$ is a functional of $f$ and $g$, and we may write
$V^{\beta,g,f}$ instead of $V^{\beta}$ whenever it needs an explicit
emphasis on its dependence of $f$ and $g$. 
It turns out that the the option price value of \eqref{eq:V} can be
obtained by a limit of a series of appropriate rebate option prices.
\begin{theorem} 
  \label{thm:conv}
  Assume (A1-A2). Suppose the rebate payoff $g$  satisfies one of 
  the following two conditions: 
  \begin{enumerate}
  \item $g(x)$ is  of sub-linear growth, i.e. $\lim_{x\to \infty}
    \frac{g(x)}{x} = 0$;  
  \item $g(x)$ is of linear growth, i.e.
    $\lim\sup_{x\to \infty} \frac{g(x)}{x} <
    \infty$, and $X^{x,t}$ is a martingale. 
  \end{enumerate}
  Then, we have the convergence for $V^\beta$ of \eqref{eq:appr1},
  $$\lim_{ \beta \to \infty} V^{\beta}(x,t)  = V(x,t).$$
  In addition, if $g \in D_\eta(\mathbb{R}^+)$ with $\gamma
  \wedge \eta <1$, then the convergence rate is the order of $1-(\gamma
  \vee \eta)$ as $\beta \to \infty$, i.e.  
  \begin{equation}
    \label{eq:rate}
    |(V - V^{\beta})(x,t)| \le K \beta^{-(1-(\gamma \vee \eta))},
    \ \forall x< \beta.
  \end{equation}
\end{theorem}

Theorem~\ref{thm:conv} shows that Monte Carlo method on the expression
$V^{\beta}$ of \eqref{eq:appr1} actually leads to correct estimation of the option price $V$, provided that the rebate function $g$ is appropriately chosen. Among the many choices of $g$, the simplest one shall be taking $g\equiv 0$. Next result summarizes the above comments.

\begin{corollary}\label{cor:mc}
Let $$V^{\beta,0} = \mathbb E_{x,t} [f(X(T)) {\bf 1}_{\{\tau^{\beta} = T}].$$
Then, $\lim_{\beta \to \infty} V^{\beta,0}(x,t) = V(x,t)$ point wisely in $x$ and $t$. Furthermore, if $f(x) = O(x^{\gamma}$ with some constant $\gamma<1$, then its convergence rate is 
$$|(V- V^{\beta})(x,t)| = O(\beta^{-1+\gamma}), \hbox{ as } \beta \to \infty.$$
\end{corollary}
In this below, we fix Monte Carlo method in Example~\ref{exm:EM} based on
convergence result in \thmref{thm:conv}.
\begin{example}
  \label{exm:EMT}
  Let's extend $f:\bar{\mathbb{R}}^+ \mapsto \mathbb{R}$ to
  $f:{\mathbb{R}} \mapsto \mathbb{R}$ by $f(x) = f(0)$ for $x<0$. With
  $X^\Delta(\cdot)$ of \eqref{eq:Xdelta}, let $\tau^\beta_\Delta$ 
  be the first hitting time of $X^\Delta$ to the barrier $\beta$.
  The modified Monte Carlo scheme to approximate $V(x,t)$ of
  Example~\ref{exm:EM} is given by,
  $$V^{\beta}_\Delta(x,t) := \mathbb{E}_{x,t} \Big[f(X^\Delta(T)) {\bf
    1}_{\{\tau^\beta_\Delta\ge T\}}\Big].
  $$
  \corref{cor:mc} implies that the rebate option price is convergent to the smallest hedging price, i.e.
  $$V^{\beta}(x,t):= \mathbb{E}_{x,t} [f(X(T)) {\bf
    1}_{\{\tau^\beta \ge T\}}]
  \to V(x,t) \quad \hbox{ as } \beta\to \infty.$$ 
  Note
  that $$X^\Delta(T)  \to X(T) \hbox{ a.s.  and } | f(X^\Delta(T)) {\bf
    1}_{\{\tau^\beta_\Delta\ge T\}}| \le  \max_{0\le x \le
    \beta} f(x).$$
  Hence, Bounded Convergence Theorem also implies that
  \begin{equation}
    \label{eq:exm:EMT1}
    V^{\beta}_\Delta(x,t) \to   V^{\beta}(x,t), \quad
    \hbox{ as } \Delta \to 0.
  \end{equation}
  As a result, the option price $V(x,t)$ of \eqref{eq:V} can be
  approximated by using Monte Carlo method on the above rebate
  options, i.e.
  $$\lim_{\beta \to \infty} \lim_{\Delta \to 0} V^{\beta}_\Delta (x,t)
  = V(x,t). $$ 
\end{example}

  Regarding the estimation by Monte Carlo method, one may take the 
  simplest choice $g(\beta) \equiv 0$ for the
  rebate payoff as of Corollary~\ref{cor:mc}, see also \cite{ELTS08p}. 
  However, Corollary~\ref{cor:mc} can not be utilized the approximation by PDE numerical method, since it may cause a discontinuity at the corner
  $(\beta, T)$  
  of the terminial-boundary datum when $f(\beta) \neq 0$.

\subsection{Approximation by PDE numerical method}

For the above Monte Carlo method on \eqref{eq:appr1}, yet another
to be mentioned is a drawback in  the computation by PDE numerical
methods due to the possible discontinuity of the boundary-terminal
data. 

To illustrate this issue, we write Black-Scholes PDE associated
to  the rebate option price $V^{\beta}(x,t)$ of \eqref{eq:appr1},
\begin{equation}
  \label{eq:pdeT}
  \left\{
    \begin{array}{ll}
      \displaystyle u_t + \frac 1 2 \sigma^2(x) u_{xx} = 0 \hbox{ on }
      Q_\beta := (0,\beta) \times (0,T);\\  
      u(x, T) = f(x), \ \forall x\in [0, \beta]; \\
      u(0,t) = f(0), \ u(\beta, t) = g(\beta),  \ \forall t\in (0,T).\\ 
    \end{array} \right.
\end{equation}
Note that,  PDE \eqref{eq:pdeT} has a
discontinuous  corner at the point $(\beta, T)$
if $g(\beta) \neq f(\beta)$. Also recall that, the choice of
$g(\beta) = f(\beta)$ may not be possible, like in CEV model of
Example~\ref{exm:cev}.  

It is well known that, if $g(\beta) \neq f(\beta)$ and the
boundary-terminal data is discontinuous, then
one can not expect the unique solution of \eqref{eq:pdeT}
continuous up to the boundary. Furthermore, 
the discontinuity and the singularity at the corner propagate the
numerical errors quickly throughout its entire domain for the
numerical PDE methods, such as finite element method (FEM) or finite
difference method (FDM), see more discussions in \cite{SYZ07} and the
references therein. Therefore, the unique solvability and the
regularity of the solution are crucial to make use of the existing PDE
numerical methods.  

To avoid this error propagation due to the discontinuity of the
boundary-terminal data, we provide an alternative choice to
\eqref{eq:appr1} by revising the terminal payoff: Consider a rebate
option of barrier $\beta$ with 
\begin{enumerate}
\item zero rebate payoff, i.e. $g(\beta) \equiv 0$;
\item and a revised terminal payoff
  \begin{equation}
    \label{eq:fbeta}
    f^\beta(x) = f(x) {\bf 1}_{\{x\le \beta/2\}} + \frac{2 f(x) (\beta -
      x)}{\beta} {\bf 1}_{\{\beta/2 <x \le \beta\}}.  
  \end{equation}
\end{enumerate}
In this case, the rebate option price 
\begin{equation}
  \label{eq:appr2}
  \widetilde V^{\beta}(x,t) = \mathbb{E}_{x,t} [ f^\beta(X(T))
  {\bf 1}_{\{\tau^\beta =  T\}}] 
\end{equation}
is associated to  PDE
\begin{equation}
  \label{eq:pdeTb}
  \left\{
    \begin{array}{ll}
      (E)_\beta & \displaystyle u_t + \frac 1 2 \sigma^2(x) u_{xx} = 0
      \hbox{ on }  Q_\beta := (0,\beta) \times (0,T);\\  
      (BD) & u(0,t) = f(0), \ u(\beta, t) = 0,  \ \forall t\in (0,T);\\
      (TD)_\beta & u(x, T) = f^\beta (x), \ \forall x\in [0, \beta].      
    \end{array} \right.
\end{equation}

Observe that, the revised terminal data $f^\beta$ not only makes the
terminal-boundary data continuous at the corner $(\beta, T)$, but also
preserves  H\"{o}lder regularity of the original terminal data $f$
regardless how large the value $\beta$ is.

Although PDE \eqref{eq:pdeTb} is degenerate at $x=0$, one can still
has unique classical solution by utilizing Shauder's interior
estimate. Also, its solution is
indeed equal to the revised rebate option price $\widetilde V^\beta$ of
\eqref{eq:appr2}, see Lemma~\ref{lem:app1}. Moreover, by using
comparison principle twice on two different truncated domains, one can
show its unique solution $\widetilde V^\beta$ must be convergent to 
the desired value $V$ of \eqref{eq:V}, 

\begin{theorem}
  \label{thm:contbd}
  Assume (A1-A2). Then, $\widetilde V^{\beta}$ of \eqref{eq:appr2} is 
  the unique $C^{2,1}(Q_\beta) \cap C(\overline Q_\beta)$ solution of
  PDE \eqref{eq:pdeTb}, and
  $$\lim_{\beta \to \infty} \widetilde V^{\beta} (x,t) = V(x,t), \
  \forall (x,t) \in Q.$$
  In addition, if $\gamma<1$ in (A2), then the convergence rate is
  $$|\widetilde V^{\beta} - V|(x,t) \le K \beta^{-1+\gamma}.$$
\end{theorem}

Thanks to the Theorem~\ref{thm:contbd}, one can use either well
established FDM or FEM on PDE \eqref{eq:pdeTb} for a large $\beta$ to
estimate the smallest superhedging price.

\section{Proof of main results}
In this section, we will first characterize the value function
$V$. Based on the properties of $V$, we can estimate $|V - V^\beta|$
to prove  \thmref{thm:conv},  and $|V - \bar V^\beta|$ to prove
\thmref{thm:contbd}, respectively.
\subsection{Characterization of the option price $V$} 
We have seen that the option price $V$ of \eqref{eq:V} is one of the
solutions of $BS(Q,f)$. To proceed, we need identify which solution
corresponds to the option price $V$ among many.  This enables us to
establish the connection between parabolic partial
differential equation $BS(Q,f)$ and probability representation
\eqref{eq:V}.

\begin{proposition}
  \label{prop:V}
  Assume (A1-A2). Then,  value function  $V$ of \eqref{eq:V} is
  \begin{enumerate}
  \item  the smallest lower-bounded   $C^{2,1}(Q)\cap
    C_\gamma(\overline Q)$ solution of  $BS(Q,f)$.
  \item the unique $C^{2,1}(Q)\cap C(\overline Q)$ solution of
    $BS(Q,f)$ if and    only if $\sigma$ satisfies
    \begin{equation}
      \label{eq:iff}
      \int_1^\infty \frac{x}{\sigma^2(x)} dx = \infty. 
    \end{equation}
  \end{enumerate}
\end{proposition}
\begin{proof}
  Theorem 3.2 of \cite{ET09} shows that $V$ is a $C^{2,1}(Q)\cap
  C(\overline Q)$ solution of  
  $BS(Q,f)$. Applying super-martingale property of $X(T)$ and Jensen's
  inequality, the next  derivation
  shows that $V\in C_\gamma(\overline Q)$,
  $$V(x,t) = \mathbb{E}_{x,t} [f(X(T))] \le K (1 +
  \mathbb{E}_{x,t}[X^\gamma(T)]) \le K (1 +x^\gamma ).  $$
  For the necessary and sufficient condition on uniqueness, we refer
  the proof to \cite{BX10}. It remains to show $V$ is the
  smallest lower bounded solution. 
  Sometimes, we use $X$ to denote $X^{x,0}$ without ambiguity in this
  proof. Note that, by path-wise uniqueness of the solution to
  \eqref{eq:sde} 
  $$Y(t) \triangleq  V(X^{x,0}(t),t) = \mathbb{E}
  [f(X^{X^{x,0}(t), t}(T)) | \mathcal{F}_t] = \mathbb{E}
  [f(X^{x,0}(T))|\mathcal{F}_t]$$  
  is a martingale process. Suppose $\hat V\in C^{2,1}(Q)\cap
  C(\overline Q)$ is an arbitrary 
  lower bounded solution of $BS(Q,f)$, then Ito's
  formula applying to $\hat Y(t) \triangleq \hat V(X(t), t)$ leads to 
  $$\hat Y(t) = V(X(0),0) + \int_0^t \hat V_x (X(s),s) \sigma(X(s))
  dW(s),$$ 
  and $\hat Y(t)$ is a lower bounded local martingale, hence  is a
  super-martingale.  Therefore, we have 
  $$\hat Y(0) \ge \mathbb{E}[\hat Y(T)] = \mathbb{E}[f(X(T))] =  Y(0)$$
  and this implies 
  $$\hat V(x,0) \ge V(x,0).$$
  We can similarly prove for   $\hat V(x,t) \ge V(x,t)$ for all $t$.
\end{proof}

\begin{proposition}\label{prop:nnV}
  Assuming (A1-A2), $BS(Q,f)$ only admits non-negative solution
  in the space of lower bounded $C^{2,1}(Q) \cap C(\overline Q)$
  functions. 
\end{proposition}
\begin{proof}
\propref{prop:V} shows that $V$ is the smallest lower-bounded solution
of PDE. Since $V\ge 0$ by definition of \eqref{eq:V}, it
implies any lower-bounded solution $u$ satisfies  $u\ge V\ge 0$.    
\end{proof}

In  \exmref{exm:cev}, we have seen that $BS(Q,f)$ of CEV model has
multiple solutions. We continue this model to demonstrate
\propref{prop:nnV}, a 
solution smaller than $V$ must  be unbounded from below.
\begin{example}
  By \propref{prop:V}, the explicit solution $V\ge 0$ of \eqref{eq:cev1}
  in CEV model  smallest lower-bounded solution of $BS(Q,f)$. In fact
  one can find, 
  $$v(x,t) = x  \Big(1 - \lambda \Phi\Big(- \frac{1}{x\sqrt{T-t}}\Big)
  \Big), \lambda>2 $$ is a smaller solution, i.e. $v\le V$ in
  $Q$. However, $v$ is not lower-bounded, i.e. $v(x,t) \to -\infty$ as
  $x\to \infty$.
\end{example}

\subsection{Proof of Theorem~\ref{thm:conv}}
Recall that the domain of the value function $V$ is given on the
domain $\overline Q$ of \eqref{eq:bsq}, and its related truncated domain 
$Q_\beta$ is given by \eqref{eq:pdeT}. 
Let $\varphi: \overline Q \to \mathbb{\bar R^+}$ be a measurable
function. We introduce the truncated value function $V^{\beta,
  \varphi}$ for convenience, 
\begin{equation}
  \label{eq:truncv}
  V^{\beta, \varphi}(x,t) = \left\{
  \begin{array}{ll}
    \mathbb{E}_{x,t} [ \varphi(X(\tau^{\beta}), \tau^{\beta})],  & \forall 
    (x,t)\in \overline Q_\beta, \\
    \varphi(x,t) & \ \hbox{ Otherwise. }
  \end{array}\right.
\end{equation}
where the stopping time $\tau^\beta$ of \eqref{eq:taubeta}  is  the
first hitting time to the barrier $\beta$. By the above definition,
$$V^{\beta,  \varphi_1} (x,t) = V^\beta(x,t)$$
for the $V^\beta$ of \eqref{eq:appr1}, if we set 
\begin{equation}
  \label{eq:phi1}
  \varphi_1 (x,t) = g(x) {\bf 1}_{\{t<T\}} + f(x) {\bf 1}_{\{t = T\}}.
\end{equation}

With the above setup, to prove Theorem~\ref{thm:conv}, our goal is to
estimate $|V^{\beta, \varphi_1} - V|$ as $\beta \to \infty$ with
$\varphi_1$ of \eqref{eq:phi1} and the constraint on $g$ given in
Theorem~\ref{thm:conv}. We emphasize here, $\varphi_1$ may not be
continuous up to the boundary, i.e. $\varphi_1 \notin C(\overline{Q})$
when $g(x)< f(x)$ for some $x>0$.

\begin{lemma}
  \label{lem:dpp}
  Assume (A1-A2). Then,
  \begin{enumerate}
  \item $V(x,t) = V^{\beta,V}(x,t)$ for all $0<x<\beta$.
  \item If  $\varphi, \psi: \overline Q \to \mathbb{\bar R^+}$
    are two measurable functions satisfying  $\varphi \ge \psi$ on
    $\partial^* Q_\beta$, 
    then $$V^{\beta, \varphi} \ge V^{\beta, \psi}, \ \forall \beta>0.$$
  \end{enumerate}
\end{lemma}
\begin{proof}
  $X^{x,t}$ is the unique strong solution of \eqref{eq:sde} due to
  (A1). Therefore, the conclusion follows from the following simple 
  derivation using   tower   property and strong Markov property:
  $$
  \begin{array}{ll}
    V(x,t) & =  \displaystyle \mathbb{E} [f(X^{x,t}(T))
    |\mathcal{F}_t]\\
    & = \displaystyle \mathbb{E} [ f(X^{x,t}(T))
    |\mathcal{F}_{\tau^\beta}] |\mathcal{F}_t ]\\
    & = \displaystyle \mathbb{E} [ \mathbb{E}[ f(X^{X(\tau^\beta),
      \tau^\beta}(T)) |\mathcal{F}_{\tau^\beta}] |\mathcal{F}_t ]\\
    & = \displaystyle \mathbb{E} [ V(X(\tau^\beta),
    \tau^\beta)|\mathcal{F}_t ]\\
    & = V^{\beta, V}(x,t).
  \end{array}
  $$
  Monotonicity of $V^{\beta, \varphi}$  in $\varphi$ follows directly
  from the definition of $V^{\beta,\varphi}$ of \eqref{eq:truncv}.
\end{proof}

It is noted that, two results of Lemma~\ref{lem:dpp} correspond to
uniqueness and comparison principle of its associated  PDE. However,
we provide the probabilistic proof Lemma~\ref{lem:dpp} , since we want
to cover potentially discontinuous function $\varphi_1$ of
\eqref{eq:phi1}, in which uniqueness may not remain true.

\begin{lemma}
  \label{lem:1}
  Assume (A1-A2) and $g\ge 0$. Then, $V^{\beta,\varphi_1}$ defined by 
  \eqref{eq:truncv} and \eqref{eq:phi1} satisfies  
  $$\lim_{\beta \to \infty} V^{\beta,\varphi_1} (x,t) \ge V(x,t).$$
  In addition, equality holds in the above if and only if
  \begin{equation}
    \label{eq:iff1}
    \lim_{\beta \to \infty} \mathbb{E}_{x,t} \Big[ g(\beta) {\bf
      1}_{\{\tau^\beta<T\}} \Big] = 0 
  \end{equation}
\end{lemma}

\begin{proof}
  We start with the following
  observation: The solution $X := X^{t,x}$ of  
  \eqref{eq:sde} does not 
  explode almost surely by \cite[5.5.3]{KS98}, i.e. 
  \begin{equation}
    \label{eq:nonexplode}
    \lim_{\beta \to \infty} \tau^{\beta} = T, \quad
    \hbox{a.s.-}\mathbb{P}
  \end{equation}
  Due to this fact together with Monotone Convergence Theorem, we
  obtain   following identities: 
  \begin{equation}
    \label{eq:lem:1:1}
    \begin{array}{ll}
      \lim_{\beta\to \infty} \mathbb{E}_{x,t} \Big[f(X(T))
      {\bf 1}_{\{\tau^\beta = T\}} \Big]  &= \mathbb{E}_{x,t} \Big[
      \lim_{\beta \to \infty} f(X(T))
      {\bf 1}_{\{\tau^\beta = T\}} \Big] \\ & = \mathbb{E}_{x,t} \Big[
      f(X(T))\Big] = V(x,t).
    \end{array}
  \end{equation}
  By the definition of $\varphi_1$ of \eqref{eq:phi1}, this results in 
  $$
  \begin{array}{ll}
    \displaystyle \lim_{\beta\to \infty} V^{\beta, \varphi_1}(x,t)
    \\  
    \displaystyle = \lim_{\beta\to
      \infty} \mathbb{E}_{x,t} [\varphi_1(X(\tau^\beta), \tau^\beta) 
    {\bf 1}_{\{\tau^\beta<T\}} ] + \lim_{\beta\to
      \infty} \mathbb{E}_{x,t} [\varphi_1(X(\tau^\beta), \tau^\beta)
    {\bf 1}_{\{\tau^\beta = T\}} ] \\
    \displaystyle = \lim_{\beta\to
      \infty} \mathbb{E}_{x,t} [\varphi_1(\beta, \tau^\beta)
    {\bf 1}_{\{\tau^\beta<T\}} ] + \lim_{\beta\to
      \infty} \mathbb{E}_{x,t} [f(X(T)) {\bf 1}_{\{\tau^\beta = T\}} ] \\
    \displaystyle = \lim_{\beta\to
      \infty} \mathbb{E}_{x,t} [g(\beta) {\bf 1}_{\{\tau^\beta<T\}} ]
    + V(x,t). 
  \end{array}
  $$
  Rearranging the above identity, we have
  \begin{equation}
    \label{eq:idval}
    V(x,t) = \lim_{\beta\to \infty} V^{\beta, \varphi_1}(x,t)  -
    \lim_{\beta\to \infty} \mathbb{E}_{x,t} [g(\beta)
    {\bf 1}_{\{\tau^\beta<T\}} ]. 
  \end{equation}
  Note that three terms in \eqref{eq:idval} are all
  non-negative. Hence, $ \lim_{\beta\to \infty} V^{\beta, \varphi_1}(x,t) \ge
  V(x,t)$ and equality  holds if and only if \eqref{eq:iff1} holds. 
\end{proof}

As mentioned in \eqref{eq:nonexplode}, the solution $ X^{x,t}$ of
\eqref{eq:sde} does not  explode almost surely, and this can be
rewritten as
$$\mathbb{P}(\tau^{x,t,\beta}<T) \to 0 \hbox{ as } \beta \to\infty.$$
An interesting question about this is that, how fast does the above
probability converge to zero? The answer to this question is indeed
useful to  obtain the convergence rate of the truncated approximation.  
\begin{proposition}
  \label{prop:mtgl}
  Fix $(x,t)\in Q$ and assume (A1-A2). As $\beta \to \infty$, stopping
  time $\tau^{x,t,\beta}$ of \eqref{eq:taubeta}  satisfies
  \begin{enumerate}
  \item $\mathbb{P}\{\tau^{x,t,\beta} < T\} = O(1/\beta)$. 
  \item   Moreover,  $ \mathbb{P}\{\tau^{x,t,\beta} <
    T\} = o(1/\beta)$ if and only if $\{X^{t,x}(s): t\le s \le T\}$ is a
    martingale .  
  \end{enumerate}
\end{proposition}
\begin{proof} By taking $g(x) = f(x) = x$ in \eqref{eq:idval},
  $$\lim_{\beta \to \infty} \mathbb{E}_{x,t}[X(\tau^\beta)] =
  \lim_{\beta \to \infty} \beta \mathbb{P}\{\tau^{x,t,\beta} < T\}
  + \mathbb{E}_{x,t} [X(T)].$$
  For all $\beta>x$, since $\{X^{x,t}(\tau^\beta\wedge s):s>t\}$ is a 
  bounded local martingale, hence it is martingale. So,
  $\mathbb{E}_{x,t}[X(\tau^\beta)] = x$ for all $\beta>x$. Rearranging
  the above identity,   we have
  \begin{equation}
    \label{eq:lprob}
    \lim_{\beta \to \infty} \beta \mathbb{P}\{\tau^{x,t,\beta} < T\}
    = x -  \mathbb{E}_{x,t} [X(T)]  
  \end{equation}
   \eqref{eq:lprob} implies
   \begin{enumerate}
   \item Since $\mathbb{E}_{x,t} [X(T)]  \ge 0$, $\lim_{\beta \to
       \infty} \beta     \mathbb{P}\{\tau^{x,t,\beta} < T\} \le x
     <\infty$, which shows 
     $\mathbb{P}\{\tau^{x,t,\beta} < T\} = O(1/\beta)$.  
   \item $\{X^{t,x}(s): t\le s \le T\}$ is a martingale  if and only if
     $x =  \mathbb{E}_{x,t} [X(T)]$, if and only if
     $\mathbb{P}\{\tau^{x,t,\beta} < T\} = o(1/\beta)$.   
   \end{enumerate}
 \end{proof}

Finally, we are now ready fo the proof of  Theorem~\ref{thm:conv}.
\begin{proof}[{\bf Proof of Theorem~\ref{thm:conv}}] We first show its
  convergence, then obtain convergence rate.
  \begin{enumerate}
  \item Regarding its convergence, it is enough to verify \eqref{eq:iff1} by
    \lemref{lem:1}.
    Note that
    $$
    \begin{array}{ll}
      \displaystyle \lim_{\beta \to \infty} \mathbb{E}_{x,t} \Big[
      g(\beta) {\bf 1}_{\{\tau^\beta<T\}} \Big]  
      \displaystyle \le
      \lim_{\beta \to \infty} \frac{ g(\beta)}{\beta} \lim_{\beta \to 
        \infty} \beta  \mathbb{P}\{\tau^{x,t,\beta} <  T\}.  
    \end{array}
    $$  
    \begin{enumerate}
    \item If $g$ is of sub-linear growth, then $  \lim_{\beta \to
        \infty} \frac{g(\beta)}{\beta} = 0$. Hence,  \eqref{eq:iff1}
      holds due to the first result of \propref{prop:mtgl};
    \item
      On the other hand, if  $X^{t,x}$ is  a martingale, then  we have  
      $\lim_{\beta \to \infty} \beta \mathbb{P}\{\tau^{x,t,\beta} <  T\} =
      0$ from the second result of \propref{prop:mtgl}, and
      \eqref{eq:iff1} remains true provided that $g$ is of linear
      growth. 
    \end{enumerate}
  \item Since $V (x,t) = V^{\beta, V}(x,t)$ for all $\beta>x$ by
    \lemref{lem:dpp}, we have the following identity:
    $$        (V - V^{\beta, \varphi_1})(x,t) = (V^{\beta, V} -
    V^{\beta, \varphi_1})(x,t)  = \mathbb{E} [ (V - \varphi_1)
    (X(\tau^\beta), \tau^\beta) {\bf 1}_{\{ \tau^\beta <T \}}]. $$
    Setting $\overline V := \sup_{t\in [0,T)} V(x,t)$, we can rewrite 
    \begin{equation}
      \label{eq:dppV}
        |(V - V^{\beta, \varphi_1})(x,t)| \le (|\overline V| + |g|) (\beta)
        \mathbb{E}  [ {\bf 1}_{\{ \tau^\beta <T \}}].
    \end{equation}
    Since $\overline V \in C_\gamma(\mathbb{R}^+)$ by \propref{prop:V}
    and $g \in D_\eta(\mathbb{R}^+)$, we have $|\overline
    {V} + g | \in  D_{\gamma\vee \eta} (\mathbb{R}^+)$. Hence, write
    \eqref{eq:dppV} by Proposition~\ref{prop:mtgl}
      $$|(V - V^{\beta, \varphi_1})(x,t)| \le  (|\overline{V}| +
      |g|)(\beta)| O(1/\beta) \le K   \beta^{(\gamma \vee \eta) - 1}, 
      $$
    which finally results in \eqref{eq:rate}.
  \end{enumerate}
\end{proof}

\subsection{Proof of Theorem~\ref{thm:contbd}}

\begin{lemma}
  \label{lem:app1}
  Assume (A1-A2). Then $\widetilde V^{\beta}$ of \eqref{eq:truncv} is 
  the unique solution of \eqref{eq:pdeTb} in the space of
  $C^{2,1}(Q_\beta)  \cap  C(\overline  Q_\beta)$. 
\end{lemma}
\begin{proof}
  Fix $(x,t) \in Q_\beta$. Take $\alpha \in (0, x/2)$. Recall 
  $Q_\beta^\alpha = Q_\beta \cap (\overline Q_\alpha)^c$  be an open
  set. Also define
  $$\tau^{\alpha,\beta} = \inf\{s>t: (X^{x,t}(s),s) \notin
  Q^\alpha_\beta\}.$$ 
  Due to the uniform ellipticity, 
  \begin{equation}
    \label{eq:Vab}
  V^{\alpha, \beta} (x,t) :=
  \mathbb{E}_{x,t} [f^\beta(X(\tau^{\alpha,\beta}))]
  \end{equation}
  is the unique classical solution of 
  \begin{equation}
    \label{eq:pdealbe}
    \left\{
      \begin{array}{ll}
        u_t + \frac 1 2 \sigma^2 (x) u_{xx} = 0,  & \hbox{ on }
        Q^\alpha_{\beta} = (\alpha,\beta)\times (0,T)
        \\
        u(\beta,t) = 0, \ u(\alpha, t) = f^\beta(\alpha), & \ \forall
        t\in (0,T) \\
        u(x,T) = f^\beta(x), & \ \forall x\in [\alpha, \beta].
      \end{array}\right.
  \end{equation}
  If we restrict $V^{\alpha, \beta}$ on the subdomain $Q^{x/2}_\beta$,
  it solves following PDE uniquely,
  \begin{equation}
    \label{eq:pdealbe1}
    \left\{
      \begin{array}{ll}
        u_t + \frac 1 2 \sigma^2 (x) u_{xx} = 0,  & \hbox{ on }
        Q^{x/2}_{\beta} = (x/2,\beta)\times (0,T)
        \\
        u(\beta,t) = 0, \ u(x/2, t) = V^{\alpha,\beta}(x/2,t), & \ \forall
        t\in (0,T) \\
        u(x,T) = f^\beta(x), & \ \forall x\in [x/2, \beta].
      \end{array}\right.
  \end{equation}
  Furthermore, by using Shauder estimate Theorem 4.9 together with
  Theorem 5.9 of \cite{Lie96}, one can have estimate on weighted
  H\"older norm, i.e.
  $$|V^{\alpha, \beta}|^*_{2.5, Q_\beta^{x/2}} \le K |V^{\alpha,
    \beta}|_{0,  Q_\beta^{x/2}}  $$
  for some constant $K$ independent to $\alpha$. On the other hand, by
  definition \eqref{eq:Vab}, we have
  $$|V^{\alpha, \beta}|_{0,  Q_\beta^{x/2}}  = \sup_{Q_\beta^{x/2}}
  V^{\alpha, \beta} \le \sup_{x\in [0, \beta]} |f^\beta(x)| \le
  \sup_{x\in [0, \beta]} |f^\beta(x)|  \le K$$
  for some $K$ independent to $\alpha$.   Let $d = \frac 1 2 \min\{x,
  \beta -x, t, T-t\}$,   which must be 
  less than the minimum distance of $x$ to any point in the parabolic 
  boundary $\partial^* Q_\beta^\alpha$. Consider a neighborhood of $x$
  given by  $N_x(r) := (x- r, x+ r) \times (t- r, t+r)$. By definition
  of the weighted norm (Page 47 of \cite{Lie96}),   we finally have
  the following $\alpha$-uniform estimate on $N_x(d)$, 
  $$|V^{\alpha, \beta}|_{2.5, N_x(d)} \le
  |V^{\alpha, \beta}|^*_{2.5, Q_\beta^{x/2}} \le K$$ 
  Therefore, Arzela-Ascoli Theorem implies that there exists a
  subsequence of $\{V^{\alpha, \beta}: \alpha \in (0,x/2)\}$,
  which is uniformly convergent to a function $u$ on
  $N_x(d)$, i.e.
  $$V^{\alpha,\beta} \to u \hbox{ as } \alpha \to 0, \ \hbox{
    uniformly on }  N_x(d).$$
  The uniform convergence implies that the limit function is $u\in
  C^{2,1}(N_x(d))$. 
  Using the facts of almost sure convergence
  $\tau^{\alpha,\beta} \to \tau^\beta$, together with dominated
  convergence theorem, one can check that 
  $$\lim_{\alpha\to 0} V^{\alpha, \beta} (x,t)
  = \mathbb{E}_{x,t} [ \lim_{\alpha \to 0}
  f^\beta (X(\tau^{\alpha,\beta}))]  = \widetilde V^{\beta}(x,t), \
  \hbox{pointwisely}.$$
  Hence, $V^{\beta} = u$ solves $(E)_\beta$  of  \eqref{eq:pdeTb} in the
  classical sense.  By bounded convergence theorem, one can also
  show $\widetilde V^\beta(x,t) \in C(\overline Q_\beta)$ from the
  facts
  \begin{equation}
    \label{eq:lim}
    \lim_{x\to \beta} \widetilde V^\beta(x,t) = 0, \
    \lim_{x\to 0} \widetilde V^\beta(x,t) = f(0), 
    \lim_{t\to T} \widetilde V^\beta(x,t) = f(x), 
  \end{equation}
  Thus, we conclude $\widetilde V^\beta$ is the classical solution of
  \eqref{eq:pdeTb}. Moreover, strong solution satisfies maximum
  principle, and hence the uniqueness follows from Corollary 2.4 of
  \cite{Lie96}.  
\end{proof}

Now, we are ready to prove  Theorem~\ref{thm:contbd}.
\begin{proof}[{\bf Proof of Theorem~\ref{thm:contbd}}] 
  $\widetilde V^\beta$ is the unique solution of
  \eqref{eq:pdeTb} by \lemref{lem:app1}. 
  Fix $(x_0,t_0) \in Q$ and $\beta>2x_0$. We will use comparison
  principle of Lemma~\ref{lem:dpp} twice to obtain the desired
  results.
  Define
  $\varphi_2: \overline Q \mapsto \mathbb{R}$ by 
  $$\varphi_2 (x,t) = f(x) {\bf 1}_{\{t = T\}}. $$
  Since $\varphi_2 (x,t) \le \widetilde V^\beta$ on
  $\partial^* Q_{\beta/2}$, we can apply Lemma~\ref{lem:dpp} on
  $\overline Q_{\beta/2}$ to obtain $V^{\beta/2, \varphi_2} (x_0,t_0) \le
  \widetilde V^\beta (x_0,t_0)$.
  Similarly, since $\widetilde V^\beta \le \varphi_2(x,t)$ on
  $\partial^* Q_{\beta}$ by its definition, we 
  apply  Lemma~\ref{lem:dpp} on
  $\overline Q_{\beta}$ to obtain $V^{\beta, \varphi_2}  (x_0,t_0) \ge
  \widetilde V^\beta (x_0,t_0)$. Thus, we have inequality 
  \begin{equation}
    \label{eq:1}
    V^{\beta/2, \varphi_2} (x_0,t_0) \le \widetilde V^{\beta} (x_0,t_0) \le
    V^{\beta, \varphi_2}(x_0,t_0). 
  \end{equation}
  Taking $\lim_{\beta\to \infty}$ in the above inequality and using
  \thmref{thm:conv}, all three terms shall converge to the  same value
  $V(x_0,t_0)$.  The rate of the convergence is the
  combined result of \eqref{eq:1} and \eqref{eq:rate}.
\end{proof}

\section{Further remarks}
This paper studies an approximation to the smallest hedging price of
European option using rebate options. 
From mathematical point of view, this work concerns on the approximation of
the value function $V$ of \eqref{eq:V} by truncating the domain $Q$
and imposing suitable Cauchy-Dirichlet data $g$. 

The main result on the convergence \thmref{thm:conv} provides that, if
the 
function $g$ is chosen to satisfy sublinear growth in $x$ uniformly in
$t\in [0,T)$, then the truncated value $V^{\beta, g}$ converges to
$V$. This enables practitioners to adopt EM methods on big enough
truncated domain $Q_\beta$ to get a close value of $V$, as demonstrated
in \exmref{exm:EMT}.  

On the other hand, to adopt numerical PDE techniques, continuous
Cauchy-Dirichlet data is desired to get a good approximation. 
However, if the payoff $f$ is given as of a linear growth, $g$ is
taken 
as of a sublinear growth in $x$ for the purpose of the convergence by 
\thmref{thm:conv}, then it's not possible to have a  continuous
solution of Black-Scholes PDE. Alternatively, we provide a continuous
Cauchy-Dirichlet data by modifying the terminal payoff appropriately.

\bibliographystyle{plain}
%\bibliography{/Users/songqsh/Documents/R/notes/refs}

\begin{thebibliography}{10}

\bibitem{BX10}
Erhan Bayraktar and Hao Xing.
\newblock On the uniqueness of classical solutions of {C}auchy problems.
\newblock {\em Proc. Amer. Math. Soc.}, 138(6):2061--2064, 2010.

\bibitem{CH05}
Alexander M.~G. Cox and David~G. Hobson.
\newblock Local martingales, bubbles and option prices.
\newblock {\em Finance Stoch.}, 9(4):477--492, 2005.

\bibitem{ELTS08p}
Erik Ekstrom, Per Lotstedt, Lina Von~Sydow, and Johan Tysk.
\newblock Numerical option pricing in the presence of bubbles.
\newblock {\em Preprint}, 2008.

\bibitem{ET09}
Erik Ekstr{\"o}m and Johan Tysk.
\newblock Bubbles, convexity and the {B}lack-{S}choles equation.
\newblock {\em Ann. Appl. Probab.}, 19(4):1369--1384, 2009.

\bibitem{FK10}
Daniel Fernholz and Ioannis Karatzas.
\newblock On optimal arbitrage.
\newblock {\em Ann. Appl. Probab.}, 20(4):1179--1204, 2010.

\bibitem{JPS07}
Robert~A. Jarrow, Philip Protter, and Kazuhiro Shimbo.
\newblock Asset price bubbles in complete markets.
\newblock In {\em Advances in mathematical finance}, Appl. Numer. Harmon.
  Anal., pages 97--121. Birkh\"auser Boston, Boston, MA, 2007.

\bibitem{JYC09}
Monique Jeanblanc, Marc Yor, and Marc Chesney.
\newblock {\em Mathematical methods for financial markets}.
\newblock Springer Finance. Springer-Verlag London Ltd., London, 2009.

\bibitem{KS98}
Ioannis Karatzas and Steven~E. Shreve.
\newblock {\em Methods of mathematical finance}, volume~39 of {\em Applications
  of Mathematics (New York)}.
\newblock Springer-Verlag, New York, 1998.

\bibitem{Lie96}
Gary~M. Lieberman.
\newblock {\em Second order parabolic differential equations}.
\newblock World Scientific Publishing Co. Inc., River Edge, NJ, 1996.

\bibitem{SYZ07}
Q.~S. Song, G.~Yin, and Z.~Zhang.
\newblock An {$\epsilon$}-uniform finite element method for singularly
  perturbed two-point boundary value problems.
\newblock {\em Int. J. Numer. Anal. Model.}, 4(1):127--140, 2007.

\bibitem{Wil07}
Paul Wilmott.
\newblock {\em Paul Wilmott Introduces Quantitative Finance}.
\newblock John Wiley \& Sons, 2nd edition, 2007.

\end{thebibliography}
%\end{document}

\def\cprime{$'$}

\end{document}